\newcounter{myctr}
\def\myitem{\refstepcounter{myctr}\bibfont\noindent\ifnum\themyctr>9\else\phantom{0}\fi\hangindent17pt\themyctr.\enskip}

\documentclass{article}
\usepackage{hyperref}
\usepackage[super,sort,compress]{cite}
\usepackage{amsfonts,paralist,mathtools,amsthm}
\usepackage{tikz}
\usetikzlibrary{matrix}
\usetikzlibrary{decorations.pathreplacing}

\newcommand{\ket}[1]{|{#1}\rangle}
\newcommand{\bra}[1]{\langle{#1}|}
\DeclareMathOperator{\Tr}{tr}
\DeclareMathOperator{\wt}{wt}

\newtheorem{theorem}{Theorem}

\newtheorem{proposition}{Proposition}
\newtheorem{corollary}{Corollary}
\newtheorem{definition}{Definition}
\newtheorem{conjecture}{Conjecture}
\newtheorem{example}{Example}

\begin{document}


\title{Encoding classical information in gauge subsystems of quantum codes}

\author{Andrew Nemec\footnote{Corresponding author.} \; and Andreas Klappenecker}


\maketitle


\begin{abstract}
We show how to construct hybrid quantum-classical codes from subsystem codes by encoding the classical information into the gauge qudits using gauge fixing. Unlike previous work on hybrid codes, we allow for two separate minimum distances, one for the quantum information and one for the classical information. We give an explicit construction of hybrid codes from two classical linear codes using Bacon-Casaccino subsystem codes, as well as several new examples of good hybrid code.
\end{abstract}



\markboth{Nemec and Klappenecker}
{Encoding classical information in gauge subsystems}

\section{Introduction}

Hybrid codes allow for the simultaneous transmission of both quantum and classical information across a quantum channel. Devetak and Shor \cite{Devetak2005} showed for certain small error rates that simultaneous transmission is superior to the time-sharing of a quantum channel, and subsequent work on the topic has been focused primarily on information-theoretic results \cite{Hsieh2010a, Hsieh2010b, Yard2005}. The first examples of finite-length hybrid codes were given by Kremsky, Hsieh, and Brun \cite{Kremsky2008} as a generalization of entanglement-assisted quantum codes, and later Grassl, Lu, and Zeng \cite{Grassl2017} gave multiple examples of good hybrid codes with small parameters using a codeword stabilizer construction. Surprisingly, these codes provide an advantage over optimal quantum codes regardless of the error rate.

Further examples of good hybrid codes were constructed by the authors \cite{Nemec2019, Nemec2020} over the Pauli channel and by Li, Lyles, and Poon \cite{Li2019} over a fully correlated quantum channel where the space of errors is spanned by $I^{\otimes n}$, $X^{\otimes n}$, $Y^{\otimes n}$, and $Z^{\otimes n}$. Additional work on hybrid codes from an operator-theoretic perspective has also been done by B{\'e}ny, Kempf, and Kribs \cite{Beny2007a, Beny2007b} and Majidy \cite{Majidy2018}. While few good hybrid code constructions are known, there are already multiple areas in which they can be used, including the protection of hybrid quantum memory \cite{Kuperberg2003} and the construction of hybrid secret sharing schemes \cite{Zhang2011}. Additionally, the work on higher rank matrical ranges by Cao et al. \cite{Cao2020} was inspired by hybrid codes.

Previous work on hybrid codes has assumed that both the quantum and classical information should be protected from all errors of weight up to the same minimum distance $d$. In this paper we introduce hybrid codes with two separate minimum distances for quantum and classical information. Loosening this restriction on the minimum distance allows us to construct new hybrid stabilizer codes by encoding classical information in the gauge qudits of subsystem codes, making use of gauge fixing. Using this result, we show how to construct hybrid codes from classical codes using Bacon-Casaccino subsystem codes \cite{Bacon2006b} including a family of good hybrid codes constructed using Bacon-Shor subsystem codes. We also give multiple additional examples of good hybrid codes, including a $\left[\!\left[9,3\!:\!1,3\right]\!\right]_{2}$ code that can encode one more qubit than the length 9 hybrid stabilizer code given by Grassl et al. \cite{Grassl2017}, as well as one derived from Kitaev's 18-qubit toric code \cite{Kitaev1997}. Finally, we conjecture that all hybrid stabilizer codes must satisfy a variant of the quantum Singleton bound. 

\subsection{Stabilizer Codes}
\label{stab}

A quantum code with parameters $\left(\!\left(n,K,d\right)\!\right)_{q}$ is a $K$-dimensional subspace $\mathcal{C}$ of a Hilbert space $\mathcal{H}=\mathbb{C}^{q^{n}}$ that can detect any errors on up to $d-1$ physical qudits. The most well-known class of quantum codes are the stabilizer codes \cite{Calderbank1998, Gottesman1997}. Stabilizer codes are the quantum analogues of classical additive codes, and we write their parameters as $\left[\!\left[n,k,d\right]\!\right]_{q}$, where $k=\log_{q}\!\left(K\right)$. While in general $K$ does not need to be an integral power of $q$, it will always be an integral power of $p$, the characteristic of the finite field $\mathbb{F}_{q}$.

Just as binary stabilizer codes are defined as the joint eigenspace of a subgroup of the $n$ qudit error group generated by tensor products of the Pauli matrices, nonbinary stabilizer codes are definied in a similar way using nice error bases \cite{Klappenecker2002, Klappenecker2003, Knill1996}. Let $\mathbb{F}_{q}$ be a finite field of characteristic $p$, where $q=p^{\ell}$. We define the trace function $\Tr:\mathbb{F}_{q}\rightarrow\mathbb{F}_{p}$ by $$\Tr\!\left(x\right)=\sum\limits_{i=0}^{\ell-1}x^{p^{i}}.$$ Let $a,b\in\mathbb{F}_{q}$ and denote by $\ket{x}$ the computational basis of $\mathbb{C}^{q}$ labeled by the elements $x\in\mathbb{F}_{q}$. The unitary operators $X\!\left(a\right)$ and $Z\!\left(b\right)$ are defined by $$X\!\left(a\right)\ket{x}=\ket{x+a} \text{ and } Z\!\left(b\right)\ket{x}=\omega^{\Tr\left(bx\right)}\ket{x},$$ where $\omega$ is the primitive $p$-th root of unity $e^{2\pi i/p}$. The set $$\mathcal{E}=\left\{X\!\left(a\right)Z\!\left(b\right)\mid a,b\in\mathbb{F}_{q}\right\}$$ forms a nice error basis for $\mathbb{C}^{q}$, modeling errors on a single qudit.

This can be extended to a system of $n$ qudits by taking tensor products of the elements of $\mathcal{E}$: let $\mathbf{a}=\left(a_{1},\dots,a_{n}\right)\in\mathbb{F}_{q}^{n}$ and define the unitary operators $X\!\left(\mathbf{a}\right)=X\!\left(a_{1}\right)\otimes\cdots\otimes X\!\left(a_{n}\right)$ and $Z\!\left(\mathbf{a}\right)=Z\!\left(a_{1}\right)\otimes\cdots\otimes Z\!\left(a_{n}\right)$. Then $$\mathcal{E}_{n}=\left\{X\!\left(\mathbf{a}\right)Z\!\left(\mathbf{b}\right)\mid \mathbf{a},\mathbf{b}\in\mathbb{F}_{q}^{n}\right\}$$ is a nice error basis for $\mathbb{C}^{q^{n}}$ and $$G_{n}=\left\{\omega^{c}X\!\left(\mathbf{a}\right)Z\!\left(\mathbf{b}\right)\mid \mathbf{a},\mathbf{b}\in\mathbb{F}_{q}^{n},c\in\mathbb{F}_{p}\right\}$$ is the error group generated by the elements of $\mathcal{E}_{n}$ (if $\mathbb{F}_{q}$ has characteristic $2$, replacing $\omega$ with $i$ and letting $c\in\mathbb{Z}_{4}$ produces the standard version of complex Pauli matrices). When $q$ is prime, the finite field $\mathbb{F}_{q}$ can be generated as an additive group by a single element, so $G_{n}=\left\langle\omega I, X_{j}\!\left(1\right), Z_{j}\!\left(1\right)\mid j\in\left[n\right]\right\rangle$, where $X_{j}\!\left(1\right)$ operates only on the $j$-th qubit. When $q$ is not a prime, but a prime power, any element in the field may be written as $a_{0}+a_{1}\alpha+\cdots+a_{\ell-1}\alpha^{\ell-1}$ where $a_{i}\in\mathbb{F}_{p}$ and $\alpha$ is a root of an irreducible polynomial in $\mathbb{F}_{p}$ of degree $\ell$. Using this we have $G_{n}=\left\langle\omega I, X_{j}\!\left(\alpha^{i}\right), Z_{j}\!\left(\alpha^{i}\right)\mid 0\leq i<\ell, j\in\left[n\right]\right\rangle$.

The weight $\wt\!\left(\cdot\right)$ of an element $E\in G_{n}$ is the number of tensor components of $E$ that are not scalar multiples of the identity matrix. Any two elements $E$ and $E'$ of $G_{n}$, where $E=\omega^{c}X\!\left(\mathbf{a}\right)Z\!\left(\mathbf{b}\right)$ and $E'=\omega^{c'}X\!\left(\mathbf{a}'\right)Z\!\left(\mathbf{b}'\right)$, satisfy the following commutation relation: \begin{equation*} EE'=\omega^{\Tr\left(\mathbf{b}\cdot\mathbf{a}'-\mathbf{b}'\cdot\mathbf{a}\right)}E'E. \end{equation*}

Let $\mathcal{S}$ be some abelian subgroup of $G_{n}$ that does not contain a scalar multiple of the identity matrix. A stabilizer code $\mathcal{C}$ is the joint $+1$-eigenspace of $\mathcal{S}$, that is $$\mathcal{C}=\bigcap\limits_{E\in\mathcal{S}}\left\{v\in\mathbb{C}^{q^{n}}\mid Ev=v\right\}.$$ The group $\mathcal{S}$ is called the stabilizer group of the code and has order $q^{n-k}$, generated by $\ell\left(n-k\right)$ elements of $G_{n}$.

The centralizer of the stabilizer group are those elements in $G_{n}$ that commute with every element of $\mathcal{S}$, which is traditionally denoted by $N\!\left(\mathcal{S}\right)$. The elements of $N\!\left(\mathcal{S}\right)/\mathcal{S}Z\!\left(G_{n}\right)$, where $\mathcal{S}Z\!\left(G_{n}\right)$ is the group generated by $\mathcal{S}$ and the center $Z\!\left(G_{n}\right)$ of the group $G_{n}$, are cosets whose elements are Pauli operators on the logical qudits. We denote the logical operators on the $i$-th logical qudit by $\overline{X_{i}\!\left(a\right)}$ and $\overline{Z_{i}\!\left(b\right)}$, with $a,b\in\mathbb{F}_{q}$. These operators are not unique, as any element in the same coset will have the same effect on the quantum code $\mathcal{C}$. The labeling of these operators is somewhat arbitrary, their only requirement being that they satisfy the following commutation and non-commutation relations that generalize the commutation and anticommutation relations of the Pauli matrices: $\left[\overline{X_{i}\!\left(a\right)},\overline{X_{j}\!\left(b\right)}\right]=0$, $\left[\overline{Z_{i}\!\left(a\right)},\overline{Z_{j}\!\left(b\right)}\right]=0$, $\left[\overline{X_{i}\!\left(a\right)},\overline{Z_{j}\!\left(b\right)}\right]=0$ if $i\neq j$, and $\overline{X_{i}\left(a\right)}\;\overline{Z_{i}\!\left(b\right)}=\omega^{\Tr\left(-b\cdot a\right)}\overline{Z_{i}\!\left(b\right)}\;\overline{X_{i}\left(a\right)}$. For example, the operators can be trivially relabeled by swapping the $X_{i}$ and $Z_{i}$ operators.

\subsection{Subsystem Codes}

Subsystem codes (also called operator quantum error-correcting codes) are a generalization of stabilizer codes that enforce a tensor product structure on the code subspace $Q=A\otimes B$. Quantum information is encoded into subsystem $A$, while subsystem $B$, known as the gauge subsystem, is useful for fault tolerance \cite{Aliferis2007} and designing improved decoding algorithms \cite{Sarvepalli2009}. However, no information is encoded in subsystem $B$, so in a certain sense it is unused space.


One way to view subsystem codes is through the stabilizer formalism of the previous section. Informally, a subsystem code can be viewed as a stabilizer code where only a subset of the logical qudits are used to encode quantum information. The logical qudits containing the quantum information correspond to the $K$-dimensional subsystem $A$, while the unused gauge qudits correpond to the $R$-dimensional subsystem $B$. Similar to stabilizer codes, we write the parameters of a subsystem code as $\left[\!\left[n,k,r,d\right]\!\right]_{q}$ where $k=\log_{q}\!\left(K\right)$ and $r=\log_{q}\!\left(R\right)$. A subsystem code has $\ell\left(n-k-r\right)$ mutually commuting generators $S_{i}$ that generate the abelian stabilizer group $\mathcal{S}$ of the subsystem code. The $KR$-dimensional subspace $Q$ is then the $+1$-eigenpace of the elements of the stabilizer group $\mathcal{S}$.

To induce the subsystem $A\otimes B$ on $Q$, we define the gauge group $\mathcal{G}$, which consists of those Pauli operators on $Q$ that act as the identity on $A$. These include elements in $\mathcal{S}$, as well as the logical operators on the subsystem $B$, which are generated by $\ell r$ pairs of gauge operators $G_{i}^{X}$ and $G_{i}^{Z}$ such that $G_{i}^{X}$ and $G_{j}^{Z}$ do not commute if $i=j$ and commute otherwise, $G_{i}^{X}$ and $G_{j}^{X}$ all commute, and $G_{i}^{Z}$ and $G_{j}^{Z}$ all commute. The gauge group is then given by \begin{equation*} \mathcal{G}=\left\langle \omega, \mathcal{S}, G_{i}^{X}, G_{i}^{Z}\mid i\in\left[\ell r\right]\right\rangle. \end{equation*} Each pair $G_{i}^{X}$ and $G_{i}^{Z}$ corresponds to some pair of logical operators $\overline{X_{i}\!\left(a\right)}$ and $\overline{Z_{i}\!\left(a\right)}$ on the stabilizer code $Q$, but they are written differently to better distinguish them from the logical operators on the subsystem $A$, which are given by $\mathcal{L}=N\!\left(\mathcal{S}\right)/\mathcal{G}$. For further details on the stabilizer formalism of subsystem codes, see Kribs and Poulin \cite{Kribs2013} and Poulin \cite{Poulin2005}.

\section{Hybrid Codes}

We now would like to simultaneously transmit a classical message along with our quantum information. A hybrid code has parameters $\left(\!\left(n,K\!:\!M,d\!:\!c\right)\!\right)_{q}$ if and only if it can simultaneously encode a superposition of $K$ orthogonal quantum states as well as one of $M$ different classical messages into the Hilbert space $\mathcal{H}=\mathbb{C}^{q^{n}}$, while detecting all errors of weight less than $d$ and $c$ on the quantum and classical information respectively. The hybrid code $\mathcal{C}$ may be thought of as a collection of $M$ orthogonal quantum codes $\mathcal{C}_{m}$ of dimension $K$, indexed by the classical message $m\in\left[M\right]=\left\{1,2,\dots,M\right\}$, as seen in Figure~\ref{hybridcollection}. We refer to the codes $\mathcal{C}_{m}$ as the inner codes and $\mathcal{C}=\left\{\mathcal{C}_{m}\mid m\in\left[M\right]\right\}$ as the outer code. To send a quantum state $\ket{\varphi}$ and a classical message $m$, we simply encode $\ket{\varphi}$ into the quantum code $\mathcal{C}_{m}$.

\begin{figure}[t]
\centering
\begin{tikzpicture}
\draw[draw=black] (0,0) rectangle ++(10,5);
\draw[draw=black] (1.5,1) rectangle ++(7,3);
\draw[draw=black,dotted] (1.5,1) rectangle ++(5.25,3);
\draw[draw=black,dotted] (1.5,1) rectangle ++(3.5,3);
\draw[draw=black,dotted] (1.5,1) rectangle ++(1.75,3);
\draw (9.25,4.5) node {$\mathcal{H}$};
\draw (1.25,1) node {$\mathcal{C}$};
\draw (2.375,2.5) node {$\mathcal{C}_{00}$};
\draw (4.175,2.5) node {$\mathcal{C}_{01}$};
\draw (5.875,2.5) node {$\mathcal{C}_{10}$};
\draw (7.675,2.5) node {$\mathcal{C}_{11}$};
\end{tikzpicture}
\caption{Each hybrid code $\mathcal{C}$ is a collection of orthogonal quantum codes $\mathcal{C}_{i}$ indexed by a classical message $i$, here represented as a binary string in $\left\{00,01,10,11\right\}$.}
\label{hybridcollection}
\end{figure}
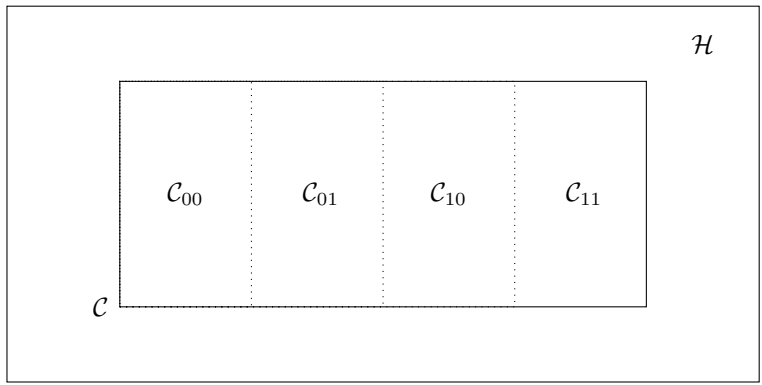

If the quantum and classical minimum distances are the same (i.e., $d=c$), we write $\left(\!\left(n,K\!:\!M,d\right)\!\right)_{q}$. If both the outer code and all of the inner codes are stabilizer codes, we refer to the code as a hybrid stabilizer code and write its parameters as $\left[\!\left[n,k\!:\!m,d\!:\!c\right]\!\right]_{q}$ where $k=\log_{q}\!\left(K\right)$ and $m=\log_{q}\!\left(M\right)$.

Grassl et al.\cite{Grassl2017} presented a set of necessary and sufficient conditions for the error-correcting capabilities of hybrid codes with $d=c$ that generalize the Knill-Laflamme conditions \cite{Knill1997} for quantum codes. Here we generalize these conditions further to allow for hybrid codes with $d\leq c$.

\begin{theorem}\label{hybridknilllaflamme}
An $\left(\!\left(n,K\!:\!M,d\!:\!c\right)\!\right)_{q}$ hybrid code with $d\leq c$ can detect up to $d-1$ errors to the quantum information and up to $c-1$ errors to the classical information if and only if
\begin{enumerate}
\item $P_{a}EP_{a}=\lambda_{E,a}P_{a}$, for all $a\in\left[M\right]$ and all error operators $E$ such that $\wt\!\left(E\right)<d$, and 
\item $P_{a}FP_{b}=0$, for all $a,b\in\left[M\right]$, $a\neq b$, and all error operators $F$ such that $\wt\!\left(F\right)<c$.
\end{enumerate}
\end{theorem}
\begin{proof}
Suppose that (1) and (2) hold. If the weight of an error $E$ on the system is less than $d$, then (1) implies that the hybrid code can detect an error on the quantum information of weight less than $d$, following directly from the Knill-Laflamme conditions for quantum codes \cite{Knill1997}. Additionally, (2) implies that the image of the codes under all the errors of weight less than $c$ are all mutually orthogonal, that is, for $a\neq b$, $P_{a}\perp\left\langle EP_{b}:\wt\!\left(E\right)<c\right\rangle$. This means that by applying a measurement based on our projectors $P_{a}$ we can always detect an error to the classical information. If instead an error $F$ with $d\leq\wt\!\left(F\right)<c$ affects the system, then we can no longer detect the error on the quantum information, but since $P_{a}FP_{b}=0$, $a\neq b$, still holds for the error $F$, we can perform a measurement and detect an error to the classical information.

Now suppose that either (1) or (2) fails to hold. If (1) fails to hold, then the Knill-Laflamme conditions tell us that there is an error to the quantum information of weight less than $d$ that the code cannot detect. If (2) fails to hold, then there is an error $F$ of weight less than $c$ such that for some $a\neq b$, $P_{a}$ and $FP_{b}$ will not be orthogonal, meaning we will not be able to completely distinguish between the two of them.
\end{proof}

As with quantum codes, an error-correction variant of the conditions immediately follows.

\begin{corollary}\label{hybridknilllaflammecor}
An $\left(\!\left(n,K\!:\!M,d\!:\!c\right)\!\right)_{q}$ hybrid code with $d\leq c$ can correct up to $\left\lfloor\frac{d-1}{2}\right\rfloor$ errors to the quantum information and up to $\left\lfloor\frac{c-1}{2}\right\rfloor$ errors to the classical information if and only if
\begin{enumerate}
\item $P_{a}E^{\dagger}FP_{a}=\lambda_{E,F,a}P_{a}$, for all $a\in\left[M\right]$ and all error operators $E,F$ such that $\wt\!\left(E\right),\wt\!\left(F\right)\leq\left\lfloor\frac{d-1}{2}\right\rfloor$, and 
\item $P_{a}E^{\dagger}FP_{b}=0$, for all $a,b\in\left[M\right]$, $a\neq b$, and all error operators $E,F$ such that $wt\!\left(E\right),\wt\!\left(F\right)\leq\left\lfloor\frac{c-1}{2}\right\rfloor$.
\end{enumerate}
\end{corollary}

Note that when $c<d$ there is a potential problem. Consider the following: let $\mathcal{C}_{a}$ and $\mathcal{C}_{b}$ be $1$-dimensional inner codes in a hybrid code with $P_{i}=\ket{\psi_{i}}\bra{\psi_{i}}$ the projector onto $\mathcal{C}_{i}$, and suppose the code satisfies conditions (1) and (2) in Theorem \ref{hybridknilllaflamme}. We may still have an error $E$ with $c\leq\wt\!\left(E\right)<d$ such that $\bra{\psi_{a}}E\ket{\psi_{b}}=\alpha\neq0$. Setting up a measurement and supposing that $\ket{\psi_{b}}$ is sent, we get \begin{align*} \left(P_{a}+P_{b}\right)E\ket{\psi_{b}} & = \ket{\psi_{a}}\bra{\psi_{a}}E\ket{\psi_{b}}+\ket{\psi_{b}}\bra{\psi_{b}}E\ket{\psi_{b}} \\ & = \alpha\ket{\psi_{a}}+\lambda_{E,b}\ket{\psi_{b}},\end{align*} which is a superposition of encoded states from the two inner codes. However, as we will show here and in Section \ref{hybridsubsystem}, we can still construct hybrid codes with $c<d$ by encoding the quantum and classical information using a subsystem structure on the encoding subspace.

Given a code $\mathcal{C}$ with a subsystem structure $A\otimes B$ on it, let $\left\{\ket{\varphi_{i}}\right\vert i\in \left[K\right]\}$ and $\left\{\ket{v_{i}}\mid i\in\left[M\right]\right\}$ be orthonormal bases for $A$ and $B$ respectively. We define the operators $$P_{a,b}=\left(\sum\limits_{i=1}^{K}\ket{\varphi_{i}}\bra{\varphi_{i}}\right)\otimes\ket{v_{a}}\bra{v_{b}},$$ which allows us to write the following error-detection conditions similar to those for subsystem codes \cite{Nielsen2007}.

\begin{theorem}\label{hysubknilllaflamme}
Let $\mathcal{C}$ be an $\left(\!\left(n,K\!:\!M,d\!:\!c\right)\!\right)_{q}$ hybrid code with a subsystem structure $A\otimes B$ on it, with $\left\{\ket{\varphi_{i}}\mid i\in\left[K\right]\right\}$ and $\left\{\ket{v_{i}}\mid i\in\left[M\right]\right\}$ as orthonormal bases for $A$ and $B$ respectively. Let $P_{a}=P_{a,a}$ be the projector onto the inner code $\mathcal{C}_{a}$. Then $\mathcal{C}$ can detect up to $d-1$ errors to the quantum information and up to $c-1$ errors to the classical information if and only if
\begin{enumerate}
\item $P_{a}EP_{b}=\lambda_{E,a,b}P_{a,b}$, for all $a,b\in\left[M\right]$ and all $E$ such that $\wt\!\left(E\right)<d$, and 
\item $P_{a}FP_{b}=0$, where $a\neq b$, for all $a,b\in\left[M\right]$, $a\neq b$, and all $F$ such that $\wt\!\left(F\right)<c$.
\end{enumerate}
\end{theorem}
\begin{proof}
The proof is the same as the proof of Theorem \ref{hybridknilllaflamme}, except we must now check the case when $c<d$. Here we will first project the code onto the subspace $\mathcal{C}$ using the projector $$P=\sum\limits_{j\in\left[M\right]}P_{j}=\left(\sum\limits_{i\in\left[K\right]}\ket{\varphi_{i}}\bra{\varphi_{i}}\right)\otimes\left(\sum\limits_{j\in\left[M\right]}\ket{v_{j}}\bra{v_{j}}\right),$$ measure the classical information in subsystem $B$ in the $\left\{\ket{v_{a}}\right\}$ basis to determine which code was sent, and then use the recovery procedure associated with that code.

Suppose that (1) and (2) are true, and let $\ket{\varphi_{a}}=\left(\ket{\varphi}\otimes\ket{v_{a}}\right)$ be the encoded state and $E$ the error on the encoded state. If $\wt\!\left(E\right)< c$, then \begin{align*}PE\ket{\varphi_{a}} & = \sum\limits_{j\in\left[M\right]}P_{j}EP_{a}\ket{\varphi_{a}} \\ & = P_{a}EP_{a}\ket{\varphi_{a}},\end{align*} by condition (2). It follows from condition (1) that \begin{align*} PE\ket{\varphi_{a}} & =  P_{a}EP_{a}\ket{\varphi_{a}} \\ & \lambda_{E,a}\ket{\varphi_{a}}.\end{align*} Performing a measurement on the subsystem $B$ will not have an effect on the encoded state and it will inform us of which code was used.

If $c\leq\wt\!\left(E\right)<d$, then by condition (1) we have \begin{align*}PE\ket{\varphi_{a}} & = \sum\limits_{j\in\left[M\right]}P_{j}EP_{a}\ket{\varphi_{a}} \\ & = \sum\limits_{j\in\left[M\right]}\lambda_{E,j,a}P_{j,a}\ket{\varphi_{a}} \\ & = \ket{\varphi}\otimes\sum\limits_{j\in\left[M\right]}\lambda_{E,j,a}\ket{v_{j}}.\end{align*} Measuring the subsystem $B$ in the $\left\{\ket{v_{i}}\right\}$ basis, we get $\lambda_{E,x,a}\left(\ket{\varphi}\otimes\ket{v_{x}}\right)/\left\lvert\lambda_{E,x,a}\right\rvert$, where $x$ may not be the original classical message. However, we are still able to detect an error to the quantum information.

The converse follows the same logic as the proof of Theorem \ref{hybridknilllaflamme}, making use of the subsystem variant of the Knill-Laflamme conditions.
\end{proof}

Similar to Theorem \ref{hybridknilllaflamme}, we immediately get the error-correction variant of Theorem \ref{hysubknilllaflamme}.

\begin{corollary}\label{hysubknilllaflammecor}
Let $\mathcal{C}$ be an $\left(\!\left(n,K\!:\!M,d\!:\!c\right)\!\right)_{q}$ hybrid code with a subsystem structure $A\otimes B$ on it, with $\left\{\ket{\varphi_{i}}\mid i\in\left[K\right]\right\}$ and $\left\{\ket{v_{i}}\mid i\in\left[M\right]\right\}$ as orthonormal bases for $A$ and $B$ respectively. Let $P_{a}=P_{a,a}$ be the projector onto the inner code $\mathcal{C}_{a}$. Then $\mathcal{C}$ can correct up to $\left\lfloor\frac{d-1}{2}\right\rfloor$ errors to the quantum information and up to $\left\lfloor\frac{c-1}{2}\right\rfloor$ errors to the classical information if and only if
\begin{enumerate}
\item $P_{a}E^{\dagger}FP_{b}=\lambda_{E,F,a,b}P_{a,b}$, for all $a,b\in\left[M\right]$ and all $E,F$ such that $\wt\!\left(E\right),\wt\!\left(F\right)\leq\left\lfloor\frac{d-1}{2}\right\rfloor$, and 
\item $P_{a}E^{\dagger}FP_{b}=0$, for all $a,b\in\left[M\right]$, $a\neq b$, and all $E,F$ such that $\wt\!\left(E\right),\wt\!\left(F\right)\leq\left\lfloor\frac{c-1}{2}\right\rfloor$.
\end{enumerate}
\end{corollary}

We leave the cases where errors to either the quantum or classical information are corrected while errors to the other are only detected for future research.

\subsection{Genuine Hybrid Codes}

Constructing hybrid codes from quantum codes is not a particularly difficult task. When $d=c$, Grassl et al. \cite{Grassl2017} gave several simple constructions of hybrid codes from quantum codes:

\begin{proposition}[Grassl et al.\cite{Grassl2017}]\label{trivhybrid}
Hybrid codes can be constructed using the following ``trivial" constructions:
\begin{enumerate}
\item Given an $\left(\!\left(n,KM,d\right)\!\right)_{q}$ quantum code of composite dimension $KM$, there exists a hybrid code with parameters $\left(\!\left(n,K\!:\!M,d\right)\!\right)_{q}$.
\item Given an $\left[\!\left[n,k\!:\!m,d\right]\!\right]_{q}$ hybrid code with $k>0$, there exists a hybrid code with parameters $\left[\!\left[n,k-1\!:\!m+1,d\right]\!\right]_{q}$.
\item Given an $\left[\!\left[n_{1},k_{1},d\right]\!\right]_{q}$ quantum code and an $\left[n_{2},m_{2},d\right]_{q}$ classical code, there exists a hybrid code with parameters $\left[\!\left[n_{1}+n_{2},k_{1}\!:\!m_{2},d\right]\!\right]_{q}$.
\end{enumerate}
\end{proposition}

We call a hybrid code with $d=c$ \emph{genuine} if there is no code constructable using Proposition \ref{trivhybrid} with the same parameters. Grassl et al. \cite{Grassl2017} showed the first examples of genuine hybrid codes, constructing multiple small-parametered hybrid codes, while the authors constructed several infinite families of genuine hybrid stabilizer codes using stabilizer pasting \cite{Nemec2019}. Note that by calling such codes ``genuine", we do not mean to imply that the hybrid codes constructed using the approaches of Proposition \ref{trivhybrid} are in any sense ``fake". Hybrid codes constructed using one of these three methods are in a sense wasting a quantum resource, in that they are transmitting classical information using space that could have been used to transmit quantum information.

Similar to the case where there is a single minimum distance, we can construct trivial hybrid codes with two minimum distances using the following construction that generalizes the third construction of Proposition \ref{trivhybrid}:
\begin{proposition}\label{badhybrid}
Given an $\left(\!\left(n_{1},K_{1},d\right)\!\right)_{q}$ quantum code and an $\left(n_{2},M_{2},c\right)_{q}$ classical code, there exists a hybrid code with parameters $\left(\!\left(n_{1}+n_{2},K_{1}\!:\!M_{2},d\!:\! c\right)\!\right)_{q}$.
\end{proposition}
\begin{proof}
Use the quantum code to encode the quantum information on the first $n_{1}$ physical qudits and use the classical code to encode the classical information on the remaining $n_{2}$ physical qudits.
\end{proof}

To generalize the first and second constructions to allow for two minimum distances, we will define a partial order on the parameters of hybrid codes to determine which codes have ``better" parameters than others:

\begin{definition}
Given two hybrid codes $\mathcal{C}$ and $\mathcal{C}'$ with parameters $\left(\!\left(n,K\!:\!M,d\!:\!c\right)\!\right)_{q}$ and $\left(\!\left(n,K'\!:\!M',d'\!:\!c'\right)\!\right)_{q}$ respectively, we say $\mathcal{C}\preceq\mathcal{C}'$ if $KM\leq K'M'$, $K\leq K'$, $d\leq d'$, and $c\leq c'$ are all true.
\end{definition}

Note that while we write $\mathcal{C}\preceq\mathcal{C}'$, we are only comparing the parameters of the codes and not the codes themselves.

\begin{proposition}
The relation $\preceq$ defines a partial order on the set of hybrid code parameters.
\end{proposition}
\begin{proof}
The reflexivity, antisymmetry, and transitivity of $\preceq$ all follow directly from the fact that $\leq$ is a partial order.
\end{proof}

If $\mathcal{C}\preceq\mathcal{C}'$, we say that $\mathcal{C}'$ has at least as good parameters as $\mathcal{C}$. Intuitively, this covers the case when $\mathcal{C}'$ has at least one parameter greater than the corresponding parameter in $\mathcal{C}$, with all other parameters being equal. For example, an $\left[\!\left[8,3,3\right]\!\right]_{2}$ quantum code has better parameters than an $\left[\!\left[8,1,3\right]\!\right]_{2}$ quantum code, as the former can encode two more logical qubits than the latter. We also give preference to codes that can transmit more quantum information if the total amount of information that can be transmitted by each code is the same. For example, we can compare the $\left[\!\left[9,1\!:\!2,3\right]\!\right]_{2}$ hybrid code of Kremsky et al. \cite{Kremsky2008} with the $\left[\!\left[9,2\!:\!2,3\right]\!\right]_{2}$ code of Grassl et al. \cite{Grassl2017}, with the latter having better parameters since it can encode one logical qubit more than the former. Similarly, we can compare both of these codes with the $\left[\!\left[9,3\!:\!1,3\right]\!\right]_{2}$ code we construct in Example \ref{9ex}, which has better parameters than both, as we can use it to construct a $\left[\!\left[9,2\!:\!2,3\right]\!\right]_{2}$ by using one of the logical qubits to transmit a classical bit. However, we cannot compare any of these three codes with the $\left[\!\left[9,1\!:\!4,3\!:\!2\right]\!\right]_{2}$ code we construct in Example \ref{9bsex}, since it transmits more total information (has a larger sum $k+m$) but has a lower classical distance.

We call a hybrid code (with $K,M>1$, although the partial order is also defined on purely quantum and classical codes) genuine if it is a maximal element in the partially ordered set and has parameters that cannot be achieved by a code constructed using Proposition \ref{badhybrid}, and we call it a genuine hybrid stabilizer code if it satisfies the same conditions on the partially ordered set induced by $\preceq$ on the subset of hybrid stabilizer codes. Intuitively, this means that a genuine hybrid code is one in which any one parameter of the code cannot be improved without sacrificing some other parameter, with the exception that we can sacrifice one bit of classical information for one qudit of quantum information. When restricted to the case with $c=d$, we recover the original definition of genuine codes.

\subsection{Hybrid Stabilizer Codes}

For the remainder of the paper we will restrict our attention to hybrid stabilizer codes, which have a particularly nice structure. Starting with a quantum stabilizer code $\mathcal{C}_{0}$ with stabilizer group $\mathcal{S}_{0}$, we choose $M$ translation operators $t_{i}$ from different cosets of $N\!\left(\mathcal{S}_{0}\right)$ in $G_{n}$ in such a way that the cosets form a group (we will always take $t_{1}$ to be the identity). The hybrid code $\mathcal{C}$ is then the union of the translated codes: $$\mathcal{C}=\bigcup\limits_{i\in\left[M\right]}t_{i}\mathcal{C}_{0}$$ 

The stabilizer generators of the inner code $\mathcal{C}_{0}$ can be divided into a quantum stabilizer $\mathcal{S}_{\mathcal{Q}}$ and a classical stabilizer $\mathcal{S}_{\mathcal{C}}$ such that $\mathcal{S}_{0}=\left\langle\mathcal{S}_{\mathcal{Q}},\mathcal{S}_{\mathcal{C}}\right\rangle$ \cite{Kremsky2008}. The quantum stabilizer $\mathcal{S}_{\mathcal{Q}}$ is the stabilizer of the outer code $\mathcal{C}$ and is generated by those generators of $\mathcal{S}_{0}$ that commute with all of the translation operators $t_{i}$. 
The classical stabilizer $\mathcal{S}_{\mathcal{C}}$ is generated by the remaining stabilizer generators of $\mathcal{S}_{0}$, each of which does not commute with at least one translation operator. We can associate each of the $\ell m$ generators $g_{i}$ of $\mathcal{S}_{\mathcal{C}}$ with an operator $\overline{Z_{j}\!\left(\alpha^{i}\right)}$, for $i\in\left\{0,\dots,\ell-1\right\}$, $j\in\left[\left\lfloor m\right\rfloor\right]$, which acts on the $j$-th virtual qudit, as well as $\overline{Z_{\left\lceil m\right\rceil}\!\left(\alpha^{i}\right)}$ for $i\in\left\{0,\dots, \ell\left(m-\left\lfloor m\right\rfloor\right)-1\right\}$ if $m$ is not a power of $q$. Similarly, we can associate each of the generators of the translation operators $\overline{X_{j}\!\left(\alpha^{i}\right)}$ for $i\in\left\{0,\dots, \ell -1\right\}$, $j\in\left[\left\lfloor m\right\rfloor\right]$, as well as $\overline{X_{\left\lceil m\right\rceil}\!\left(\alpha^{i}\right)}$ for $i\in\left\{0,\dots, \ell\left(m-\left\lfloor m\right\rfloor\right)-1\right\}$ if $m$ is not a power of $q$. These operators satisfy the commutation relations from Section \ref{stab}, and we can associate each classical message $\mathbf{a}\in\mathbb{F}_{q}^{\left\lceil m\right\rceil}$ with the translation operator $t_{\mathbf{a}}=\overline{X\!\left(\mathbf{a}\right)}=\overline{X_{1}\!\left(a_{1}\right)}\cdot\overline{X_{2}\!\left(a_{2}\right)}\cdots\overline{X_{\left\lceil m\right\rceil}\!\left(a_{\left\lceil m \right\rceil}\right)}$. In addition to mapping between the inner codes, these translation operators are also logical operators for the outer code $\mathcal{C}$. 

The quantum and classical stabilizers are sufficient to fully define a hybrid code. The following result was originally given in the binary case by Kremsky et al. \cite{Kremsky2008} and by the authors in the case of prime fields \cite{Nemec2019}. Here we generalize it to arbitrary finite fields.

\begin{theorem} \label{phaseconstruction}
Let $\mathcal{C}$ be an $\left[\!\left[n,k\!:\!m,d\!:\!c\right]\!\right]_{q}$ hybrid stabilizer code over a finite field of characteristic $p$, where $q=p^{\ell}$, with quantum stabilizer $\mathcal{S}_{\mathcal{Q}}$ and classical stabilizer $\mathcal{S}_{\mathcal{C}}=\left\langle g_{i}\mid i\in\left[\ell m\right]\right\rangle$, where $g_{i}=\overline{Z\!\left(\mathbf{b}_{i}\right)}$, $\mathbf{b}_{i}\in\mathbb{F}_{q}^{\left\lceil m \right\rceil}$. Then the inner stabilizer code $t_{\mathbf{a}}\mathcal{C}_{0}$ associated with the classical message $\mathbf{a}\in\mathbb{F}_{q}^{\left\lceil m\right\rceil}$ is stabilized by $$\left\langle\mathcal{S}_{\mathcal{Q}},\omega^{-\Tr\left(\mathbf{b}_{i}\cdot\mathbf{a}\right)}g_{i}\mid i\in\left[\ell m\right]\right\rangle,$$ where $\omega$ is a primitive $p$-th root of unity.
\end{theorem}
\begin{proof}
Let $\ket{\varphi}$ be an encoded state of $\mathcal{C}_{0}$, so that $t_{\mathbf{a}}\ket{\varphi}$ is an encoded state of $t_{\mathbf{a}}\mathcal{C}_{0}$. Since elements of the quantum stabilizer commute with $t_{\mathbf{a}}$ and stabilize $\ket{\varphi}$, they are all elements of the stabilizer of $t_{\mathbf{a}}\mathcal{C}_{0}$. In the case of $\omega^{-\Tr\left(\mathbf{b_{i}}\cdot\mathbf{a}\right)}g_{i}$, it follows from the commutation relations that $t_{\mathbf{a}}\ket{\varphi}$ is one of its $+1$-eigenstates, so it is also in the stabilizer of $t_{\mathbf{a}}\mathcal{C}_{0}$.
\end{proof}

\section{Hybrid Codes from Subsystem Codes}\label{hybridsubsystem}

In this section we show how every subsystem code leads to a hybrid code with the same quantum error-correcting properties. While the tensor structure of classical-quantum systems (see Devetak and Shor \cite{Devetak2005} and B{\'e}ny et al. \cite{Beny2007b}) suggests that subsystem codes might be useful in constructing hybrid codes, it is not immediately obvious whether or not they can protect the encoded classical information from errors. The main idea behind our construction is to follow the reasoning of Theorem \ref{hysubknilllaflamme} and encode the quantum information in the subsystem $A$ stabilized by the stabilizer group $\mathcal{S}$, and then use \emph{gauge fixing} to encode the classical information into the subsystem $B$.

\subsection{Gauge Fixing Construction}

Gauge fixing is a technique that takes commuting gauge operators of the subsystem code and uses them to generate a larger stabilizer group $\mathcal{S}_{0}$. In essence, we are taking a subset of the gauge qudits and fixing them to certain states. Since the states are fixed, no information can be encoded on those qudits, but any errors that occur on them is now either a pure error or in the stabilizer $\mathcal{S}_{0}$.

Gauge fixing is well known in quantum error-correction for its use in code switching \cite{Bombin2015, Paetznick2013}, which allows for a way around Eastin and Knill's famous no-go theorem in fault tolerance \cite{Eastin2009}. Our construction picks a commuting set of $\ell r$ independent gauge operators of the subsystem code, and then multiplies them by a phase, which forces the gauge qudits to change to a different fixed state. For example, in a binary subsystem code if the gauge operator $G_{i}^{Z}$ is fixed, it means that the $i$-th gauge qubit will be fixed as the $+1$ eigenstate of the operator, so we have a logical $\ket{0}$ that is fixed. If instead we fix the operator $-G_{i}^{Z}$, the $i$-th gauge qubit will be fixed as $\ket{1}$, the $-1$ eigenstate of the operator.

\begin{theorem}\label{hybridconstruction}
Let $\mathcal{C}$ be an $\left[\!\left[n,k,r,d\right]\!\right]_{q}$ subsystem code. Then there exists an $\left[\!\left[n,k\!:\!r,d\!:\!c\right]\!\right]_{q}$ hybrid code.
\end{theorem}
\begin{proof}
Let $\mathcal{S}=\left\langle S_{i}\right\rangle$ be the stabilizer group of $\mathcal{C}$, which will be the stabilizer of the hybrid code's outer code. Choose $2\ell r$ operators $G_{i}^{X}$ and $G_{i}^{Z}$ where $i\in\left[\ell r\right]$, so that $$\mathcal{G}=\left\langle\omega,\mathcal{S},G_{i}^{X},G_{i}^{Z}\mid i\in\left[\ell r)\right]\right\rangle.$$ Without loss of generality, we will fix a gauge and let $$\mathcal{S}_{0}=\left\langle \mathcal{S}, G_{i}^{Z}\mid i\in\left[\ell r\right]\right\rangle$$ be the stabilizer of our inner stabilizer code $\mathcal{C}_{0}$.

The centralizer of $\mathcal{S}$ and $\mathcal{S}_{0}$ are given by $$N\!\left(\mathcal{S}\right)=\left\langle\omega, \mathcal{S}, G_{i}^{X}, G_{i}^{Z}, \overline{X_{j}}, \overline{Z_{j}}\mid i\in\left[\ell r\right], j\in\left[\ell k\right]\right\rangle$$ and $$N\!\left(\mathcal{S}_{0}\right)=\left\langle\omega, \mathcal{S}_{0}, \overline{X_{i}}, \overline{Z_{i}}\mid i\in\left[\ell k\right]\right\rangle$$ respectively. The quantum minimum distance of the hybrid code is the minimum weight of one of the logical operators on the quantum information, so it will be the identical to the minimum distance of the subsystem code, given by $d=\wt\!\left(N\!\left(\mathcal{S}\right)\setminus\mathcal{G}\right)$.

The classical minimum distance $c$ is given by the minimum weight of a logical operator on the classical information, so $c=\wt\!\left(N\!\left(\mathcal{S}\right)\setminus N\!\left(S_{0}\right)\right)$. For any two elements $t_{a},t_{b}\notin N\!\left(\mathcal{S}\right)$, $t_{a}\mathcal{C}_{0}$ and $t_{b}\mathcal{C}_{0}$ will be orthogonal to each other if and only if $t_{a}$ and $t_{b}$ are in different cosets of $N\!\left(\mathcal{S}_{0}\right)$. We will use the gauge operators $G_{i}^{X}$ to construct our translation operators as in Theorem \ref{phaseconstruction}. Any error element of the error group $G_{n}$ may be written (modulo a global phase) as $E=RSTUV$, where $R\in\mathcal{S}$ is an element of the quantum stabilizer, and $S$, $T$, $U$, and $V$ are coset representatives of the classical stabilizer $\mathcal{S}_{0}/\mathcal{S}$, the logical quantum operators $N\!\left(\mathcal{S}_{0}\right)/\mathcal{S}_{0}$, the logical classical or the translation operators $N\!\left(\mathcal{S}\right)/N\!\left(\mathcal{S}_{0}\right)$, and the pure errors $G_{n}/N\!\left(\mathcal{S}\right)$ respectively. We now have three cases to consider: (i) $\wt\!\left(E\right)<c,d$, (ii) $c\leq\wt\!\left(E\right)<d$, and (iii) $d\leq\wt\!\left(E\right)<c$:
\begin{enumerate}[(i)]
\item Suppose $\wt\!\left(E\right)<c\leq d$. Then $E\notin N\!\left(\mathcal{S}\right)\setminus\mathcal{G}$ and $E\notin N\!\left(\mathcal{S}\right)\setminus N\!\left(\mathcal{S}_{0}\right)$, meaning that any error is of the form $RSV$. If $V$ is not the identity, then the error can be detected, but if not then the error has no effect on either the quantum or classical information.
\item Suppose $c\leq\wt\!\left(E\right)<d$. Then $E\notin N\!\left(\mathcal{S}\right)\setminus\mathcal{G}$, so any error is of the form $RSUV$. If $V$ is not the identity then the error can be detected, but if not then the classical information may be corrupted. However, the quantum information will be preserved.
\item Suppose $d\leq\wt\!\left(E\right)<c$. Then $E\notin N\!\left(\mathcal{S}\right)\setminus N\!\left(\mathcal{S}_{0}\right)$, so any error is of the form $RSTV$. As in (ii), if $V$ is not the identity then the error can be detected, but if not then the quantum information may be corrupted, while the classical information will be preserved.
\end{enumerate}
Therefore the hybrid code is able to detect all errors in the quantum and classical information less than their respective minimum distances.
\end{proof}

\begin{figure}[t]
\centering
\begin{tikzpicture}
\node at (0,0) (z1) {$Z_{1}$};
\node at (0,-.5) (z2) {$Z_{2}$};
\node at (0,-1.0) (z3) {$Z_{3}$};
\node at (0,-1.5) (z4) {$Z_{4}$};
\node at (0,-2.0) (z5) {$Z_{5}$};
\node at (0,-2.5) (z6) {$Z_{6}$};
\node at (0,-3.0) (x6) {$X_{6}$};
\node at (0,-3.5) (x5) {$X_{5}$};
\node at (0,-4.0) (x4) {$X_{4}$};
\node at (0,-4.5) (x3) {$X_{3}$};
\node at (0,-5.0) (x2) {$X_{2}$};
\node at (0,-5.5) (x1) {$X_{1}$};
\draw [thick, decorate,decoration={brace, mirror, amplitude=4pt},xshift=-10pt,yshift=0pt] (0,.25) -- (0,-1.75) node [black,midway,xshift= -15pt] {$\mathcal{S}$};
\draw [thick, decorate,decoration={brace, mirror, amplitude=4pt},xshift=-10pt,yshift=0pt] (0,-2.25) -- (0,-3.25) node [black,midway,xshift= -15pt] {$\mathcal{L}$};
\draw [thick, decorate,decoration={brace, amplitude=4pt},xshift=10pt,yshift=0pt] (0,-3.25) -- (0,-3.75);
\draw [thick, decorate,decoration={brace, mirror, amplitude=4pt},xshift=-10pt,yshift=0pt] (0,-3.75) -- (0,-5.75) node [black,midway,xshift= -37pt] {Pure Errors};
\draw [thick, decorate,decoration={brace, amplitude=4pt},xshift=10pt,yshift=0pt] (0,.25) -- (0,-2.25);
\draw [thick, decorate,decoration={brace, amplitude=4pt},xshift=15pt,yshift=0pt] (0,-1.0) -- (0,-3.5) node [black,midway,xshift= 15pt] {$\mathcal{G}$};
\end{tikzpicture}
\centering
\begin{tikzpicture}
\node at (0,0) (z1) {$Z_{1}$};
\node at (0,-.5) (z2) {$Z_{2}$};
\node at (0,-1.0) (z3) {$Z_{3}$};
\node at (0,-1.5) (z4) {$Z_{4}$};
\node at (0,-2.0) (z5) {$Z_{5}$};
\node at (0,-2.5) (z6) {$Z_{6}$};
\node at (0,-3.0) (x6) {$X_{6}$};
\node at (0,-3.5) (x5) {$X_{5}$};
\node at (0,-4.0) (x4) {$X_{4}$};
\node at (0,-4.5) (x3) {$X_{3}$};
\node at (0,-5.0) (x2) {$X_{2}$};
\node at (0,-5.5) (x1) {$X_{1}$};
\draw [thick, decorate,decoration={brace, mirror, amplitude=4pt},xshift=-10pt,yshift=0pt] (0,.25) -- (0,-1.75) node [black,midway,xshift= -15pt] {$\mathcal{S}_{\mathcal{Q}}$};
\draw [thick,decorate,decoration={brace, mirror, amplitude=4pt},xshift=-10pt,yshift=0pt] (0,-1.75) -- (0,-2.25) node [black,midway,xshift= -15pt] {$\mathcal{S}_{\mathcal{C}}$};
\draw [thick, decorate,decoration={brace, mirror, amplitude=4pt},xshift=-10pt,yshift=0pt] (0,-2.25) -- (0,-3.25) node [black,midway,xshift= -43pt] {Logical Q. Ops.};
\draw [thick, decorate,decoration={brace, mirror, amplitude=4pt},xshift=-10pt,yshift=0pt] (0,-3.25) -- (0,-3.75) node [black,midway,xshift= -43pt] {Logical C. Ops.};
\draw [thick, decorate,decoration={brace, mirror, amplitude=4pt},xshift=-10pt,yshift=0pt] (0,-3.75) -- (0,-5.75) node [black,midway,xshift= -37pt] {Pure Errors};
\draw [thick, decorate,decoration={brace,mirror,amplitude=4pt},xshift=-35pt,yshift=0pt] (0,.25) -- (0,-2.25) node [black,midway,xshift= -15pt] {$\mathcal{S}_{0}$};
\draw [thick, decorate,decoration={brace, amplitude=4pt},xshift=10pt,yshift=0pt] (0,.25) -- (0,-3.25) node [black,midway,xshift= 22pt] {$N\!\left(\mathcal{S}_{0}\right)$};
\draw [thick, decorate,decoration={brace,amplitude=4pt},xshift=47pt,yshift=0pt] (0,.25) -- (0,-3.75) node [black,midway,xshift= 22pt] {$N\!\left(\mathcal{S}\right)$};
\end{tikzpicture}
\caption{The relationship between a 6 qubit subsystem code (left) and the hybrid stabilizer code (right) derived from it, such as the one given in Example \ref{6ex}. In the hybrid code the translation operators are the logical classical operators and $\mathcal{S}=\mathcal{S}_{\mathcal{Q}}$.}
\label{paulitower}
\end{figure}

From the proof, we can see that encoding the classical message in the phases of the classical stabilizer generators that occurs in Theorem \ref{phaseconstruction} is in effect gauge fixing. The relationship between the stabilizer and gauge groups of the original subsystem code and the quantum and classical stabilizer groups and the translation operators of the hybrid code are shown in Figure \ref{paulitower}.

Additionally, since all hybrid stabilizer codes may be written as a subsystem code, they may all be obtained using this construction. This allows us to make use of results for subsystem codes and apply them to hybrid stabilizer codes. For instance, in \cite{Klappenecker2007} Klappenecker and Sarvepalli showed that any $\mathbb{F}_{q}$-linear Clifford subsystem code satisfies the quantum Singleton bound, and it is conjectured that any subsystem code satisfies the bound \cite{Aly2009,Klappenecker2007}. We extend this conjecture to hybrid stabilizer codes:

\begin{conjecture} \label{singletonconjecture}
An $\left[\!\left[n,k\!:\!m,d\!:\!c\right]\!\right]_{q}$ hybrid stabilizer code satisfies the following variant of the (quantum) Singleton bound: \begin{equation*} k+m\leq n-2\left(d-1\right). \end{equation*}
\end{conjecture}

\subsection{Examples of New Hybrid Codes}

We now give several examples of new hybrid codes constructed from subsystem codes using Theorem \ref{hybridconstruction}.

\begin{example}
Using the 6-qubit subsystem code was given by Shaw et al. \cite{Shaw2008} and the construction detailed in Theorem \ref{hybridconstruction}, we get a $\left[\!\left[6,1\!:\!1,3\!:\!2\right]\!\right]_{2}$ hybrid code with the following generators:
\begin{equation*}
\label{6ex}
\left(\mkern-5mu
\begin{tikzpicture}[baseline=-.5ex]
\matrix[
  matrix of math nodes,
  column sep=.25ex, row sep=-.25ex
] (m)
{
Y & I & Z & X & X & Y \\
Z & X & I & I & X & Z \\
I & Z & X & X & X & X \\
Z & Z & Z & I & Z & I \\
I & I & I & X & I & I \\
Z & I & X & I & X & I \\
I & Z & I & I & Z & Z \\
I & I & I & Z & I & Z \\
};
\draw[line width=1pt, line cap=round, dash pattern=on 0pt off 2\pgflinewidth]
  ([yshift=.2ex] m-4-1.south west) -- ([yshift=.2ex] m-4-6.south east);
\draw[line width=.5pt]
  ([yshift=.2ex] m-5-1.south west) -- ([yshift=.2ex] m-5-6.south east);
\draw[line width=.5pt]
  ([yshift=.25ex] m-7-1.south west) -- ([yshift=.25ex] m-7-6.south east);
\draw[line width=.5pt]
  ( m-7-1.south west) -- ( m-7-6.south east);
\end{tikzpicture}\mkern-5mu
\right).
\end{equation*}

Here the stabilizer generators of the subsystem code are given above the dotted line and the gauge operator $G^{Z}$ is directly below it, so that the Pauli elements above the single solid line define the inner code $\mathcal{C}_{0}$. The logical operators on the quantum information are below the single solid line, while the logical operator on the classical information, i.e., the translation operator that takes $\mathcal{C}_{0}$ to $\mathcal{C}_{1}$ and vice versa, is given below the double solid line.

Each individual single-qubit error has a distinct syndrome, with the exception of $Y_{4}$ (the Pauli-$Y$ on the 4th qubit), $Z_{4}$, and $Z_{6}$, which all share the same syndrome. The errors $Z_{4}$ and $Z_{6}$ each map the codeword to an orthogonal subspace, so the quantum information remains unaffected, but it is impossible to determine the classical information as there are two elements of weight 2 in the outer code's centralizer, although the presence of an error on the classical information can be detected. Since $X_{4}$ is in the stabilizer of the code, the error $Y_{4}$ may be viewed as the same as $Z_{4}$.

By using both the quantum and classical Singleton bounds, we find that there cannot be any hybrid code with equivalent parameters constructed from Proposition \ref{badhybrid}. Since the linear programming bounds for hybrid stabilizer codes \cite{Grassl2017} rule out the existence of a $\left[\!\left[6,1\!:\!1,3\right]\!\right]_{2}$ code, this code is in fact genuine and saturates the bound of Conjecture \ref{singletonconjecture}.
\end{example}

In addition to creating hybrid codes with two distinct minimum distances, the construction given in Theorem \ref{hybridconstruction} can also produce genuine hybrid codes with $c=d$. In particular, the construction can take subsystem codes that appear to be at first glance quite poorly designed and produce optimal hybrid codes from them.

\begin{example}\label{9ex}
We now construct a $\left[\!\left[9,3\!:\!1,3\right]\!\right]_{2}$ hybrid stabilizer code. Starting with Gottesman's pure 8-qubit code \cite{Gottesman1996}, we extend it \cite[Lemma 69]{Ketkar2006} to an impure $\left[\!\left[9,3,3\right]\!\right]_{2}$ code. This code can be viewed as a subsystem code with gauge operators $IIIIIIIIX$ and $IIIIIIIIZ$. By our construction above, this gives a $\left[\!\left[9,3\!:\!1,3\!:\!1\right]\!\right]_{2}$ hybrid code, which cannot even detect a single error to the classical information. However, we can slightly alter the quantum stabilizer to improve this distance by appending $X$ to the end of one of the stabilizers, creating a different subsystem code. It is from this code that we construct our hybrid code.
\begin{equation*}
\left(\mkern-5mu
\begin{tikzpicture}[baseline=-.5ex]
\matrix[
  matrix of math nodes,
  column sep=.25ex, row sep=-.25ex
] (m)
{
X & X & X & X & X & X & X & X & I \\
Z & Z & Z & Z & Z & Z & Z & Z & I \\
X & I & X & I & Z & Y & Z & Y & I \\
X & I & Y & Z & X & I & Y & Z & I \\
X & Z & I & Y & I & Y & X & Z & X \\
I & I & I & I & I & I & I & I & X \\
};
\draw[line width=1pt, line cap=round, dash pattern=on 0pt off 2\pgflinewidth]
  ([yshift=.2ex] m-5-1.south west) -- ([yshift=.2ex] m-5-9.south east);
\end{tikzpicture}\mkern-5mu
\right).
\end{equation*}

Since the altered generator is an element of the inner code's stabilizer, we have not altered the inner code's stabilizer. However, the outer code has been altered in such a way that there are no weight 1 or 2 elements except for $IIIIIIIIX$, which is in the inner code's stabilizer. Therefore, by the construction we now have a $\left[\!\left[9,3\!:\!1,3\right]\!\right]_{2}$ hybrid stabilizer code. As mentioned above, this code has better parameters than both the $\left[\!\left[9,1\!:\!2,3\right]\!\right]_{2}$ code of Kremsky et al. \cite{Kremsky2008} and the $\left[\!\left[9,2\!:\!2,3\right]\!\right]_{2}$ code of Grassl et al. \cite{Grassl2017}, and since it meets the linear programming bounds for hybrid stabilizer bounds it is a genuine hybrid stabilizer code.
\end{example}

\begin{example}
We now show how to construct a hybrid code out of Kitaev's well known $\left[\!\left[18,2,3\right]\!\right]_{2}$ toric code \cite{Kitaev1997} which can be converted into an $\left[\!\left[18,2,12,3\right]\!\right]_{2}$ subsystem code similar to the way used by Poulin to convert Shor's 9-qubit code into a subsystem code \cite{Poulin2005}. Using the construction from Theorem \ref{hybridconstruction}, we can construct an $\left[\!\left[18,2\!:\!12,3\!:\!2\right]\!\right]_{2}$ hybrid code.

\begin{equation*}
\label{toric}
\left(\mkern-5mu
\begin{tikzpicture}[baseline=-.5ex]
\matrix[
  matrix of math nodes,
  column sep=.25ex, row sep=-.25ex
] (m)
{
X & X & I & X & X & I & X & X & I & X & X & X & X & X & X & I & I & I \\
I & X & X & I & X & X & I & X & X & I & I & I & X & X & X & X & X & X \\
Z & Z & Z & Z & Z & Z & I & I & I & Z & Z & I & Z & Z & I & Z & Z & I \\
I & I & I & Z & Z & Z & Z & Z & Z & I & Z & Z & I & Z & Z & I & Z & Z \\
X & I & X & I & I & I & I & I & I & X & I & I & I & I & I & X & I & I \\
X & X & I & I & I & I & I & I & I & I & X & I & I & I & I & I & X & I \\
I & I & I & X & I & X & I & I & I & X & I & I & X & I & I & I & I & I \\
I & I & I & X & X & I & I & I & I & I & X & I & I & X & I & I & I & I \\
X & X & I & X & X & I & X & X & I & I & I & I & I & I & I & I & I & I \\
I & X & X & I & X & X & I & X & X & I & I & I & I & I & I & I & I & I \\
Z & I & I & Z & I & I & I & I & I & Z & Z & I & I & I & I & I & I & I \\
I & Z & I & I & Z & I & I & I & I & I & Z & Z & I & I & I & I & I & I \\
I & I & I & Z & I & I & Z & I & I & I & I & I & Z & Z & I & I & I & I \\
I & I & I & I & Z & I & I & Z & I & I & I & I & I & Z & Z & I & I & I \\
I & I & I & I & I & I & I & I & I & Z & Z & I & Z & Z & I & Z & Z & I \\
I & I & I & I & I & I & I & I & I & I & Z & Z & I & Z & Z & I & Z & Z \\
};
\draw[line width=1pt, line cap=round, dash pattern=on 0pt off 2\pgflinewidth]
  ([yshift=.2ex] m-4-1.south west) -- ([yshift=.2ex] m-4-18.south east);
\end{tikzpicture}\mkern-5mu
\right).
\end{equation*}

Using both the classical and quantum Singleton bounds together, we can see that a hybrid code with these parameters cannot be constructed from a pair of quantum and classical codes using Proposition \ref{badhybrid}. This code also saturates the bound of Conjecture \ref{singletonconjecture}.
\end{example}

All of the previous examples are of hybrid codes with $d=3$, but the construction can be used on codes with higher minimum distances.

\begin{example}
 Here we give an example of a $\left[\!\left[12,1\!:\!1,5\!:\!4\right]\!\right]_{2}$ hybrid code, constructed by modifying the extended $\left[\!\left[12,1,5\right]\!\right]_{2}$ stabilizer code from Grassl's online table of quantum codes \cite{GrasslONLINE} in a similar manner as in Example \ref{9ex}.

\begin{equation*}
\left(\mkern-5mu
\begin{tikzpicture}[baseline=-.5ex]
\matrix[
  matrix of math nodes,
  column sep=.25ex, row sep=-.25ex
] (m)
{
X & Z & I & Z & I & X & I & Z & Z & I & I & X \\
I & Y & I & Z & Z & Y & I & I & Z & Z & I & X \\
I & Z & X & I & Z & X & I & I & I & Z & Z & X \\
I & Z & Z & Y & I & Y & Z & I & I & Z& I & I \\
I & I & Z & Z & X & X & Z & Z & I & Z & Z & I \\
I & I & Z & Z & I & I & Y & Z & Z & I & Y & I \\
I & Z & I & Z & Z & Z & Z & Y & I & Z & Y & I \\
I & I & I & Z & I & Z & I & Z & X & Z & X & I \\
I & Z & I & Z & I & I & Z & I & Z& X & X & I \\
Z & Z & Z & Z & Z & Z & I & I & I & I & I & I \\
I & I & I & I & I & I & I & I & I & I & I & X \\
};
\draw[line width=1pt, line cap=round, dash pattern=on 0pt off 2\pgflinewidth]
  ([yshift=.2ex] m-10-1.south west) -- ([yshift=.2ex] m-10-12.south east);
\end{tikzpicture}\mkern-5mu
\right).
\end{equation*}

Since a code with these parameters cannot be constructed using Proposition \ref{badhybrid}, and no $\left[\!\left[12,2,5\right]\!\right]_{2}$ code exists, this is a good code, though whether or not it is genuine depends on the existence of better codes not ruled out by bounds on hybrid codes such as linear programming bounds\cite{Grassl2017,Nemec2019}.

\end{example}

\section{Bacon-Casaccino Hybrid Codes}

We give an explicit construction of hybrid codes using the Bacon-Casaccino family of subsystem codes. This family was introduced in the binary case by Bacon and Casaccino\cite{Bacon2006b} and by Klappenecker and Sarvepalli \cite{Klappenecker2007} in the nonbinary case as a generalization of the Bacon-Shor subsystem codes \cite{Bacon2006a, Shor1995}, and allow for the construction of subsystem codes from pairs of classical linear codes that need not be self-orthogonal. For completeness, we give the result below:

\begin{theorem}[Bacon-Casaccino Codes \cite{Bacon2006b, Klappenecker2007}]\label{bccode}
For $i\in\left\{1,2\right\}$, let $C_{i}\subseteq\mathbb{F}_{q}^{n_{i}}$ be an $\mathbb{F}_{q}$-linear code with parameters $\left[n_{i},k_{i},d_{i}\right]_{q}$. Then there exists a subsystem code with the parameters
\begin{equation*}
\left[\!\left[n_{1}n_{2}, k_{1}k_{2}, \left(n_{1}-k_{1}\right)\!\left(n_{2}-k_{2}\right), \min\!\left\{d_{1}, d_{2}\right\}\right]\!\right]_{q},
\end{equation*}
that is pure to $d_{p}=\min\!\left\{d_{1}^{\perp}, d_{2}^{\perp}\right\}$, where $d_{i}^{\perp}$ denotes the minimum distance of $C_{i}^{\perp}$.
\end{theorem}

A subsystem code is said to be \emph{pure to} $d_{p}$ if its gauge group contains no error of weight less than $d_{p}$.

We give a brief explanation of this construction, restricting ourselves to the binary case for simplicity. Denote by $P_{1}$ and $P_{2}$ the parity-check matrices and $G_{1}$ and $G_{2}$ the generator matrices for the classical linear codes $C_{1}$ and $C_{2}$ respectively. We can use the rows of $P_{1}$ to define $n_{1}-k_{1}$ stabilizers $S_{i}=\otimes_{j=1}^{n_{1}}Z^{\left(P_{1}\right)_{ij}}$ of length $n_{1}$, and the stabilizer group of the code is $\mathcal{S}_{1}=\left\langle S_{1},\dots, S_{n_{1}-k_{1}}\right\rangle$, which defines a classical stabilizer code able to detect $d_{1}-1$ Pauli-$X$ errors. Similarly, we can use $P_{2}$ to define $n_{2}-k_{2}$ stabilizers $T_{i}=\otimes_{j=1}^{n_{2}}X^{\left(P_{2}\right)_{ij}}$ that generate the stabilizer group $\mathcal{S}_{2}$. This defines a classical stabilizer code able to detect $d_{2}-1$ Pauli-$Z$ errors, but here the codewords are given in the Hadamard basis $\left\{\ket{+}, \ket{-}\right\}$ rather than the computational basis $\left\{\ket{0}, \ket{1}\right\}$. By classical stabilizer code, we mean a stabilizer code in which the encoded basis states are protected against noise, but a superposition of the encoded basis states are not.

To construct a quantum subsystem code out of these two classical stabilizer codes, we arrange $n_{1}n_{2}$ qubits on an $n_{1}\times n_{2}$ rectangular lattice. We use the stabilizers from $\mathcal{S}_{1}$ to operate on each column, that is each column has a copy of $\mathcal{S}_{1}$ acting on it, and likewise those stabilizers from $\mathcal{S}_{2}$ on the rows. Let $\mathcal{T}_{1}$ be the abelian group generated by $\mathcal{S}_{1}$ acting on the columns and $\mathcal{T}_{2}$ the abelian group generated by $\mathcal{S}_{2}$ acting on the rows. The group $\mathcal{T}=\left\langle\mathcal{T}_{1}, \mathcal{T}_{2}\right\rangle$ is nonabelian, but we can construct an abelian subgroup of $\mathcal{T}$ that commutes with every element in $\mathcal{T}$ using the following construction: take an element $S\in\mathcal{S}_{1}$ and a codeword $v\in C_{2}$, and construct an element of $\mathcal{T}_{1}$ where $S^{v_{j}}$ acts on column $j$. In addition to commuting with all of the elements of $\mathcal{T}_{1}$, every element of this form also commutes with all of the elements of $\mathcal{T}_{2}$. Likewise we can construct elements in $\mathcal{T}_{2}$ that commute with all elements in $\mathcal{T}$. Together, these elements generate the stabilizer group $\mathcal{S}$ of the subsystem code.

\begin{example}\label{9bscode}
As an example we present the 9-qubit Bacon-Shor code, a subsystem code version of the original 9-qubit Shor code. Start with $C_{1}=C_{2}$ as the length 3 repetition code with generator matrix $G$ and parity-check matrix $P$ given by
\begin{equation*}
G=\begin{pmatrix}
1 & 1 & 1
\end{pmatrix}
\text{ and }
P=\begin{pmatrix}
1 & 1 & 0 \\
0 & 1 & 1
\end{pmatrix}.
\end{equation*}
Using the construction, we find that the stabilizer of the code is given by
\begin{equation*}
\mathcal{S}=\left\langle
\begin{tikzpicture}[baseline=-.5ex]
\matrix (m)[matrix of math nodes, nodes in empty cells, ] {
Z & Z & Z \\
Z & Z & Z \\
I & I & I \\
} ;
\draw (m-3-1.south west) rectangle (m-1-3.north east);
\end{tikzpicture},
\begin{tikzpicture}[baseline=-.5ex]
\matrix (m)[matrix of math nodes, nodes in empty cells, ] {
I & I & I \\
Z & Z & Z \\
Z & Z & Z \\
} ;
\draw (m-3-1.south west) rectangle (m-1-3.north east);
\end{tikzpicture},
\begin{tikzpicture}[baseline=-.5ex]
\matrix (m)[matrix of math nodes, nodes in empty cells, ] {
X & X & I \\
X & X & I \\
X & X & I \\
} ;
\draw (m-3-1.south west) rectangle (m-1-3.north east);
\end{tikzpicture},
\begin{tikzpicture}[baseline=-.5ex]
\matrix (m)[matrix of math nodes, nodes in empty cells, ] {
I & X & X \\
I & X & X \\
I & X & X \\
} ;
\draw (m-3-1.south west) rectangle (m-1-3.north east);
\end{tikzpicture}
\right\rangle.
\end{equation*}
Here the stabilizers appear as they would on the $3\times3$ lattice of qubits.

We can also choose our gauge operators in such a way so that they form four anticommuting pairs $\left(G_{i}^{Z}, G_{i}^{X}\right)$:
\begin{align*}
\left(
\begin{tikzpicture}[baseline=-.5ex]
\matrix (m)[matrix of math nodes, nodes in empty cells,ampersand replacement=\& ] {
Z \& I \& I \\
Z \& I \& I \\
I \& I \& I \\
} ;
\draw (m-3-1.south west) rectangle (m-1-3.north east);
\end{tikzpicture},
\begin{tikzpicture}[baseline=-.5ex]
\matrix (m)[matrix of math nodes, nodes in empty cells, ampersand replacement=\&] {
X \& I \& X \\
I \& I \& I \\
I \& I \& I \\
} ;
\draw (m-3-1.south west) rectangle (m-1-3.north east);
\end{tikzpicture}
\right), & \left(
\begin{tikzpicture}[baseline=-.5ex]
\matrix (m)[matrix of math nodes, nodes in empty cells, ampersand replacement=\&] {
I \& I \& I \\
Z \& I \& I \\
Z \& I \& I \\
} ;
\draw (m-3-1.south west) rectangle (m-1-3.north east);
\end{tikzpicture},
\begin{tikzpicture}[baseline=-.5ex]
\matrix (m)[matrix of math nodes, nodes in empty cells, ampersand replacement=\&] {
I \& I \& I \\
I \& I \& I \\
X \& I \& X \\
} ;
\draw (m-3-1.south west) rectangle (m-1-3.north east);
\end{tikzpicture}
\right), \\ \left(
\begin{tikzpicture}[baseline=-.5ex]
\matrix (m)[matrix of math nodes, nodes in empty cells,ampersand replacement=\& ] {
I \& Z \& I \\
I \& Z \& I \\
I \& I \& I \\
} ;
\draw (m-3-1.south west) rectangle (m-1-3.north east);
\end{tikzpicture},
\begin{tikzpicture}[baseline=-.5ex]
\matrix (m)[matrix of math nodes, nodes in empty cells,ampersand replacement=\& ] {
I \& X \& X \\
I \& I \& I \\
I \& I \& I \\
} ;
\draw (m-3-1.south west) rectangle (m-1-3.north east);
\end{tikzpicture}
\right), & \left(
\begin{tikzpicture}[baseline=-.5ex]
\matrix (m)[matrix of math nodes, nodes in empty cells,ampersand replacement=\& ] {
I \& I \& I \\
I \& Z \& I \\
I \& Z \& I \\
} ;
\draw (m-3-1.south west) rectangle (m-1-3.north east);
\end{tikzpicture},
\begin{tikzpicture}[baseline=-.5ex]
\matrix (m)[matrix of math nodes, nodes in empty cells, ampersand replacement=\&] {
I \& I \& I \\
I \& I \& I \\
I \& X \& X \\
} ;
\draw (m-3-1.south west) rectangle (m-1-3.north east);
\end{tikzpicture}
\right). \\
\end{align*}
Note that if we pick a gauge and look at the subspace stabilized by the abelian group $\left\langle \mathcal{S}, G_{1}^{Z}, G_{2}^{Z}, G_{3}^{Z}, G_{4}^{Z}\right\rangle$, it is the same as the original 9-qubit Shor code (up to a permutation of the qubits).
\end{example}

We can apply our hybrid code construction from Theorem \ref{hybridconstruction} to the Bacon-Casaccino subsystem codes to construct hybrid codes out of a pair of linear codes.

\begin{theorem}
For $i\in\left\{1,2\right\}$, let $C_{i}\subseteq\mathbb{F}_{q}^{n_{i}}$ be an $\mathbb{F}_{q}$-linear code with parameters $\left[n_{i},k_{i},d_{i}\right]_{q}$. Then there exists a hybrid code with the parameters
\begin{equation*}
\left[\!\left[n_{1}n_{2}, k_{1}k_{2}\!:\!\left(n_{1}-k_{1}\right)\!\left(n_{2}-k_{2}\right), d\!:\!c\right]\!\right]_{q},
\end{equation*}
where $d=\min\!\left\{d_{1}, d_{2}\right\}$, $c\geq\min\!\left\{d,\max\!\left\{d_{1}^{\perp}, d_{2}^{\perp}\right\}\right\}$, and $d_{i}^{\perp}$ denotes the minimum distance of $C_{i}^{\perp}$.
\end{theorem}
\begin{proof}
Without loss of generality, suppose that $d_{2}^{\perp}\geq d_{1}^{\perp}$. Using Theorems \ref{hybridconstruction} and \ref{bccode}, we construct a hybrid code, gauge fixing all of the $G_{i}^{Z\left(a\right)}$ operators. The only thing that needs to be checked is the classical distance $c=\wt\!\left(N\!\left(\mathcal{S}\right)\setminus N\!\left(\mathcal{S}_{0}\right)\right)$. Since all of the translation operators are tensor products of $X$-type operators and the identity matrix, we only need to consider the minimum distance of operators of this type. Note that it may be possible to do better than this by picking both $G_{i}^{Z}$ and $G_{j}^{X}$ operators that commute to be fixed. Suppose that $d_{2}^{\perp}\leq d$. Then $\mathcal{G}$ does not contain any $X$-type operators of weight less than $d_{2}^{\perp}$, so $N\!\left(\mathcal{S}\right)$ also does not contain any $X$-type operators of weight less than $d_{2}^{\perp}$, giving us the lower bound $c\geq d_{2}^{\perp}$. If $d\leq d_{2}^{\perp}$, there may be an element of $\left(N\!\left(\mathcal{S}\right)\setminus N\!\left(\mathcal{S}_{0}\right)\right)\setminus\mathcal{G}$ of weight less than $d_{2}^{\perp}$, since a logical quantum $X$-type operator and a translation operator together might have a weight less than each operator separately. However, this weight will still be lower bounded by $d$. Following the same argument with $d_{1}^{\perp}\geq d_{2}^{\perp}$ gives us the lower bound on the classical minimum distance.
\end{proof}

\begin{example}\label{9bsex}
We will continue to use the 9-qubit Bacon-Shor subsystem code from Example \ref{9bscode} and show how to turn it into a $\left[\!\left[9,1\!:\!4,3\!:\!2\right]\!\right]_{2}$ hybrid code. Gauge fix the subsystem code by letting $\left\langle \mathcal{S}, G_{1}^{Z}, G_{2}^{Z}, G_{3}^{Z}, G_{4}^{Z}\right\rangle$ be the stabilizer of the code $\mathcal{C}_{0}$. To send the classical binary message $m=m_{1}m_{2}m_{3}m_{4}$, use $\left(G_{1}^{X}\right)^{m_{1}}\!\left(G_{2}^{X}\right)^{m_{2}}\!\left(G_{3}^{X}\right)^{m_{3}}\!\left(G_{4}^{X}\right)^{m_{4}}$ as the translation operator.

Similar to the previous examples, we can use the quantum and classical Singleton bounds to show that this code is superior to any hybrid code constructed using Proposition \ref{badhybrid}.
\end{example}

As mentioned above, this code cannot be compared to any of the other length 9 hybrid codes mentioned in this paper, but it does have the distinction of being able to transmit the conjectured maximal amount of total information, as it achieves the bound in Conjecture \ref{singletonconjecture}. In fact, this property is shared by all of the hybrid Bacon-Shor codes:

\begin{corollary}
Hybrid codes with parameters $$\left[\!\!\!\:\left[n^{2},1\!:\!\left(n-1\right)^{2},n\!:\!2\right]\!\!\!\:\right]_{\!\!\:2}$$ can be constructed from Bacon-Shor subsystem codes, and no code constructed using Proposition \ref{badhybrid} can have these parameters. Furthermore, these codes saturate the bound given in Conjecture \ref{singletonconjecture}.
\end{corollary}

\section{Conclusion}

In this paper we have shown how to encode classical information in the gauge qudits of subsystem codes, allowing us to use previously unused logical qudits to transmit information. The hybrid codes that arise from this construction are allowed to have separate minimum distances for the quantum and classical information. We give several examples of good hybrid codes using this construction on subsystem codes including a new genuine $\left[\!\left[9,3\!:\!1,3\right]\!\right]_{2}$ code, as well as use the Bacon-Casaccino subsystem code construction to construct hybrid stabilizer codes from a pair of classical codes and their duals, including the Bacon-Shor hybrid codes constructed from the classical repetition codes. We also conjecture that hybrid stabilizer codes must satisfy a variant of the quantum Singleton bound, which follows from a similar conjecture for subsystem codes.

Previous work on hybrid codes required the construction of good families of degenerate quantum codes to construct families of genuine hybrid codes. By relating hybrid codes to the well-studied class of subsystem codes and separating the quantum and classical minimum distances of the code, it should be easier to find families of genuine hybrid codes. One important question raised by having separate minimum distances is that of bounds for hybrid codes when $c\neq d$. In Conjecture \ref{singletonconjecture}, the variant of the quantum Singleton bound does not put any restrictions on the classical distance, so finding bounds such as the linear programming bounds for hybrid codes \cite{Grassl2017,Nemec2019} that put restrictions on both minimum distances would allow for a better understanding of these codes. Other topics of future research include the cases where errors to either the quantum or classical information are corrected while errors to the other are only detected, as we only considered the cases where errors were either both detected or both corrected.

\section*{Acknowledgments}

This research was supported in part by a Texas A\&M University T3 grant. The authors would like to thank Markus Grassl for pointing out an error in the original formulation of Theorem \ref{hybridknilllaflamme}.

\end{document}